\documentclass[12pt, reqno]{amsart}
\usepackage{graphicx,subfigure}
\usepackage{hyperref}
\usepackage{mcode}
\usepackage[section]{placeins}

\numberwithin{equation}{section}
\numberwithin{figure}{section}
\numberwithin{table}{section}

\theoremstyle{plain}
\newtheorem{Thm}{Theorem}
\newtheorem{Prop}{Proposition}

\theoremstyle{definition}

\newtheorem{example}{Example}

\begin{document}

\title[Analysis of an Epidemic Model]{Geometric Singular Perturbation Theory Analysis of an Epidemic Model with Spontaneous Human Behavioral Change}

\author[Schecter]{Stephen Schecter}
\address{
Department of Mathematics \\
North Carolina State University \\
Box 8205 \\
Raleigh, NC 27695 USA 
}
\email{schecter@ncsu.edu}
\date{June 21, 2020}

\begin{abstract}
We consider a model due to Piero Poletti and collaborators that adds spontaneous human behavioral change to the standard SIR epidemic model.   In its simplest form, the Poletti model adds one differential equation, motivated by evolutionary game theory, to the SIR model.  The new equation describes the evolution of a variable $x$ that represents the fraction of the population using normal behavior.  The remaining fraction $1-x$ uses altered behavior such as staying home, social isolation, mask wearing, etc.  Normal behavior offers a higher payoff when the number of infectives is low; altered behavior offers a higher payoff when the number is high.  We show that the entry-exit function of geometric singular perturbation theory can be used to analyze the model in the limit in which behavior changes on a much faster time scale than that of the epidemic.  In particular, behavior does not change as soon as a different behavior has a higher payoff; current behavior is sticky.  The delay until behavior changes in predicted by the entry-exit function.
\end{abstract}

\subjclass{92D30, 34E15, 91A22}

\keywords{epidemiological modeling, entry-exit function, geometric singular perturbation theory, imitation dynamics, evolutionary game theory}

\thanks{I thank Dan Marchesin and Marlon Michael Lopes Flores of IMPA for introducing me to epidemiological modeling. Their work on epidemiology is supported by FAPERJ and Instituto Serrapilheira.}

\maketitle

\section{Introduction}

A disease epidemic in a human population, such as measles, influenza, or covid-19, spreads due to a combination of pathogen characteristics and human behavior.  Pathogen characteristics determine the circumstances under which an infected person can readily infect another.  Human behavior determines how frequently those circumstances occur.

Baseline human behavior varies with the society.  In a city, crowded conditions in housing, public transportation, schools and workplaces may lead to frequent close human interactions; in a rural area this may be less true.  In East Asia mask-wearing in public is fairly common in normal circumstances; in other parts of the world it is rare.  

During an epidemic, human behavior may change due to government policies closing schools and businesses, requiring people to stay at home, and encouraging social distancing and mask-wearing.

Spontaneous changes in human behavior also affect the course of an epidemic.  People may react to an epidemic, or to information presented to them, by spontaneously reducing social contacts, staying home to the extent possible, adopting more stringent hygiene or social distancing, or wearing a mask.  People may adopt these behaviors independent of government policies; and, to the extent that they feel motivated to adopt such behaviors, they are more likely to comply with government orders and encouragement to do so.

Similarly, when an epidemic wanes, or when people are presented with information that an epidemic is waning or that the disease is less dangerous than originally feared, people may spontaneously return to normal behavior.  If restrictive government policies are still in place, compliance may decline.

In the simplest epidemic models, SIR models, the transmissibility of a disease in captured in a single parameter, $\beta$, defined as the number of ``adequate contacts" per unit time that an infected person has with other people \cite{hethcote09}.  If these other people are susceptible to the disease (not immune due to previous infection and not currently infected), an adequate contact results in a new infected individual.  The basic reproduction number of the disease, $R_0$, is $\beta$ times the typical length of time that an infected person remains infective.  If $R_0>1$, then initially, when the susceptible fraction of the population is close to 1, the number of infected individuals will grow.

Epidemic control measures aim to reduce $\beta$ by enforcing or encouraging changes in behavior.  To determine what measures to institute, governments rely on epidemic models that estimate $\beta$ under normal circumstances and under various restrictive policies.

A weakness of all epidemic models in current use, as far as I know, is that they ignore spontaneous behavioral change.  For example, the Imperial College covid-19 model \cite{imp-covid1}, which influenced the United Kingdom and United States government to institute social distancing measures \cite{wpost1}, was based on a very detailed 2006 influenza epidemic model by the same group \cite{ferguson06}.  According to the 2006 paper, 
``We do not assume any spontaneous change in the behaviour of uninfected individuals as the pandemic progresses, but note that behavioural changes that increased social distance together with some school and workplace closure occurred in past pandemics \ldots and might be likely to occur in a future pandemic even if not part of official policy. \ldots Such spontaneous changes in population behaviour might more easily reduce peak daily case incidence." 

Epidemiologists appear to be well aware that spontaneous behavioral change should be incorporated in models.  There is a fairly extensive literature on ways to do it; a review article is \cite{verelst16}.  There does not appear to be agreement on what modeling approach is best.  This  probably should not be regarded as a serious problem; a variety of different models are commonly used in epidemiology.  A more serious issue is that there has been little work on how to determine the values of the parameters in the models \cite{verelst16}.  Without approximate values for the parameters, models of spontaneous behavioral change can only yield qualitative predictions.

The goal of this paper is not to deal with the various issues of how best to account for behavioral change in epidemic models.  Instead we want to direct attention to a particular approach to behavioral change, due to  Piero Poletti and collaborators, in its simplest form \cite{poletti09, poletti10, poletti12}.  This model adds one equation, motivated by evolutionary game theory \cite{nowak-sigmund04, hofbauer-sigmund03}, to the standard SIR model.  Our goal is to show how the entry-exit function \cite{demaesschalck08} of geometric singular perturbation theory \cite{jones94, kuehn15} can be used to analyze this model.  Given values for the parameters, the entry-exit function enables one to make precise predictions in the limit where behavioral change occurs on a much faster time scale than the epidemic itself. 

To my knowledge, there have been two previous uses of the entry-exit function in epidemiological models, \cite{liu16} and \cite{kuehn20}.

Figure \ref{fig:poletti} shows a typical simulation of the Poletti model. There are three variables.  Two, $S$ and $I$, are the familiar susceptible and infective population fractions from the SIR model.  The third variable, $x$, represents the fraction of the population using normal behavior.  When $x=1$, in this simulation, the model reduces to an SIR model with $R_0 = 3$.  When $x=0$, the entire population uses altered behavior; in this simulation, the model reduces to an SIR model with $R_0 = .6$.  In the simulation, behavior changes on a time scale 200 times faster than that of the epidemic itself.  Thus, if the time scale for the epidemic is days, 1000 time units in the simulation equals five days.  The simulation shows 20,000 fast time units, or 100 days.  

\begin{figure}[htb]
\includegraphics[width=4in]{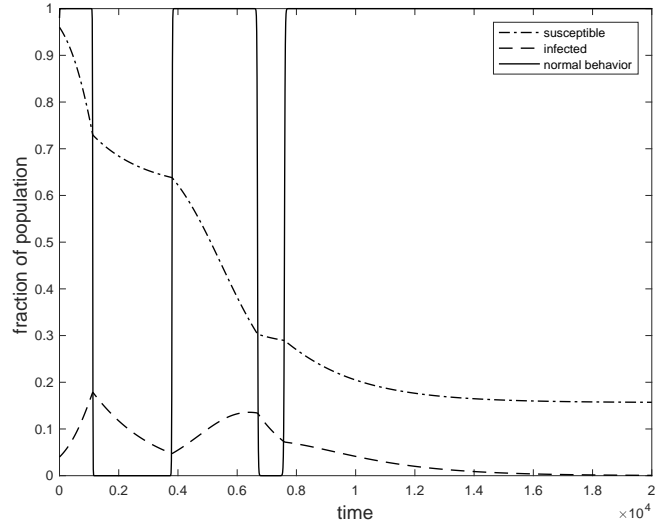}
\caption{A simulation of the Poletti model.  At the start $(S,I,x)=(.96,.04,.98)$.  Since $I<1$, almost all the population quickly adopts normal behavior.  After $I$ rises to about .18 (showing behavior stickiness), the population switches to altered behavior.  $I$ falls to about .05 (again showing behavior stickiness); the population returns to normal behavior; and $I$ rises to about .13 (second wave).  After two more behavioral switches, the epidemic dies out.}
\label{fig:poletti}
\end{figure}

Altered behavior yields a negative payoff due to loss of income, loss of social interactions, and so on.  However, altered behavior reduces the chance of getting the disease.  In this simulation, normal behavior yields a higher payoff to the individual when  $I< .1$.  When $I>.1$, altered behavior yields a higher payoff.  There is therefore a tendency to adopt altered behavior, which moderates the epidemic, when $I$ passes $.1$.  When $I$ falls below $.1$, there is a tendency to resume normal behavior.  Resuming normal behavior can result in a ``second wave" of infections, as seen in the simulation.  

In the Poletti model, behavior changes because of encounters with other people whose behavior offers a higher payoff than one's own.  Thus, in the simulation, behavior does not change immediately when $I$ passes .1; the current behavior is ``sticky."  The delay until behavior changes can be calculated in the limit from the entry-exit function.

The rather fast evolution of the epidemic in the simulation is due to the choices $R_0 = 3$ and $R_0 = .6$.

The Poletti model, in my view, plays a role similar to the SIR model: it gives the essence of the situation, stripped of complications, and can form the basis for more realistic models.  I expect that geometric singular perturbation theory will also prove useful in analyzing more realistic extensions of the model.

In the next few sections of the paper we review the SIR model (section \ref{sec:SIR}), introduce the Poletti model (section \ref{sec:Pol}), and describe and exploit the model's slow-fast structure (section \ref{sec:slow-fast}). The main result of the paper, Theorem \ref{th:smalleps}, is stated at the end of section \ref{sec:slow-fast}.  We then provide examples (section \ref{sec:ex}) and proofs (section \ref{sec:proofs}), and conclude with a brief discussion (section \ref{sec:disc}).

\section{The SIR model\label{sec:SIR}}

The Poletti model is based on the standard SIR model for an epidemic,
\begin{align}
\dot S &= -\beta SI, \label{SIR1} \\
\dot I &= \beta SI  - \gamma I, \label{SIR2} \\
\dot R &= \gamma I \label{SIR3},
\end{align}
with $\dot{\;} = \frac{d\;}{dt}$.  The variables $S$, $I$, and $R$ are population fractions; they sum to 1.  (Since $\dot S+\dot I+\dot R=0$, the sum $S+I+R$ is constant.) $S$ is the fraction of the population that is susceptible to acquiring the disease; $I$ is the fraction that is currently infected; $R$ is the fraction that has recovered.  ($R$ is sometimes called the fraction removed. If there are deaths due to the disease, they are included in $R$ without change to the model.)    Thus the equation for $R$ can be ignored; $R$ can be recovered from $R=1-S-I$.  The system reduces to
\begin{align}
\dot S &= -\beta SI, \label{SIRr1} \\
\dot I &= \beta SI  - \gamma I \label{SIRr2}
\end{align}
on the triangle $T = \{(S,I) : S\ge0, \;  I\ge0, \; S+I\le1\}$.

Let $\hat{T}= \{(S,I) \in T : S>0 \mbox{ and } I>0$.  In  $\hat{T}$ the orbits of \eqref{SIRr1}--\eqref{SIRr2} satisfy the differential equation $\frac{dI}{dS}=-1+\frac{\gamma}{\beta S}$, so they are curves 
\begin{equation}
\label{SIRorbit}
I+S-\frac{\gamma}{\beta}\ln S=C.
\end{equation}

The parameter $\beta$ was discussed in the introduction.  The average length of time an individual is infected is $\frac{1}{\gamma}$.   Thus the basic reproduction number of the disease $R_0$ mentioned in the introduction is $\frac{\beta}{\gamma}$.

Phase portraits on $T$ in the cases $0<R_0<1$ and $R_0>1$ are shown in Figure \ref{fig:SIR}.  In both cases the system has the line segment of equilibria $I=0$, $0\le S \le1$.  Each solution approaches one of the equilibria $(S_f,0)$.  In other words, when the epidemic ends, no one is infected, and $R=1-S_f$ is the fraction of the population that contracted the disease in the course of the epidemic.  

In $\hat{T}$, $\dot S<0$, so the number of susceptibles steadily falls.  If $0<R_0<1$, $\dot I<0$ in $\hat{T}$ as well, so the number of infectives also steadily falls.  If $R_0>1$, $\dot I<0$ (resp. $\dot I>0$) for $0<S<\frac{\gamma}{\beta}$ (resp. $\frac{\gamma}{\beta}<S<1)$.  Thus if a solution starts with $\frac{\gamma}{\beta}<S<1$, then $I$ increases until $S$ has fallen to $\frac{\beta}{\gamma}$; after that $I$ decreases.

\begin{figure}[htb]
\centering
\subfigure[]
{
\label{fig:SIR:a}
\includegraphics[width=3in]{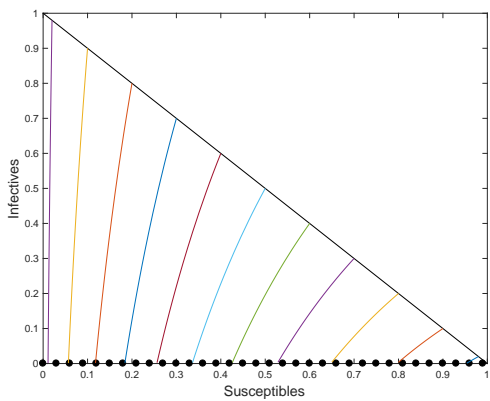}
}
\hspace{.0001in}
\subfigure[]
{
\label{fig:SIR:b}
\includegraphics[width=3in]{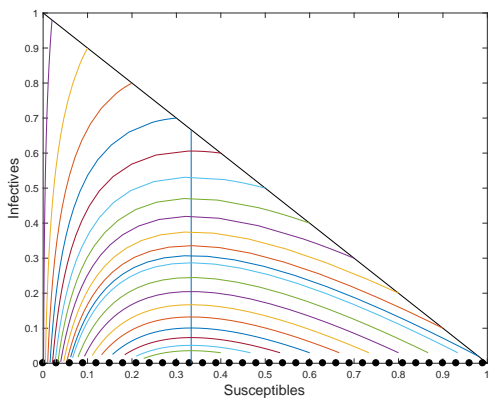}
}
\caption{Phase portraits of SIR models. (a) $\beta = 1/10$, $\gamma=1/6$, so $R_0=6/10$.  (b) $\beta = 1/2$, $\gamma=1/6$, so $R_0=3$. In case (b), the vertical line $S = \gamma/\beta = 1/3$ at which solutions attain their maximum value of $I$ is also shown.  All solutions move to the left as time increases.}
\label{fig:SIR}
\end{figure}

\section{The Poletti Model\label{sec:Pol}}

In the Poletti model, susceptible individuals have two available behaviors, normal, for which $\beta=\beta_n$ with basic reproduction number $R_0=\frac{\beta_n}{\gamma}>1$, and altered, for which $\beta=\beta_a$ with basic reproduction number $R_0=\frac{\beta_a}{\gamma}<1$.  Altered behavior may include staying home to the extent possible, practicing social distancing, mask wearing, etc.

Each behavior has a payoff to a susceptible who adopts it.  The payoffs are
$$
p_n = -m_n I \mbox{ and } p_a = -k-m_a I
$$
with $m_n$, $m_a$ and $k$ positive and $m_n>m_a$.  The negative payoff $-m_n I$ is due to the possibility that a susceptible with normal behavior will contract the disease; it is proportional to $I$, the fraction of infectives in the population.  The negative payoff $-m_a I$ is due to the possibility that a susceptible with altered behavior will contract the disease; it is also proportional to $I$, but the proportionality constant is less negative.  In addition, altered behavior has a negative payoff $-k$ independent of $I$ that represents loss of income, loss of valued social interactions, etc.  The payoff from altered behavior is higher if and only if $I>\frac{k}{m_n-m_a}$, i.e., if and only if the fraction of infectives in the population is sufficiently high.  We assume
$$
\frac{k}{m_n-m_a}<1.
$$
This assumption allows altered behavior to sometimes have a higher payoff.

Susceptibles are assumed to change their behavior from normal to altered, or vice-versa, due to imitation of other susceptibles they encounter who are using the opposite behavior and experiencing a higher payoff.  The mathematical formulation of this notion is called imitation dynamics and comes from evolutionary game theory \cite{hofbauer-sigmund03}. 

Let $x$ denote the fraction of the susceptibles using normal behavior, so that $1-x$ is the fraction using altered behavior.  We continue to let $S$ and $I$ denote the susceptible and infected fractions of the population. Then the complete model is 
\begin{align}
\dot S &= -\big(\beta_nx+\beta_a(1-x)\big)SI, \label{P1} \\
\dot I &= \big(\beta_nx+\beta_a(1-x)\big)SI - \gamma I, \label{P2} \\
\dot x &=x(1-x)(\beta_a-\beta_n)I+\frac{1}{\epsilon}x(1-x)\big(k-(m_n-m_a)I\big), \label{P3} 
\end{align}
with $\dot{\;} = \frac{d\;}{dt}$.  The state space is the prism 
$$
P=\{(S,I,x) : S\ge0, \;  I\ge0, \; S+I\le1, \; 0\le x \le 1\}. 
$$
There is also an equation for the recovered fraction of the population $R$, $\dot R = \gamma I$; we ignore it since $R$ can be recovered from $S=1-R-I$.

For the derivation of the model, see \cite{poletti09, poletti10}.  It can be intuitively understood as follows.  

The equations for $\dot S$ and $\dot I$ come from assuming that both susceptibles with normal behavior and susceptibles with altered behavior satisfy SIR models.  

The first summand in the equation for $\dot x$ is negative; it expresses the fact that susceptibles with normal behavior acquire the disease more easily than susceptibles with altered behavior, and hence more readily leave the susceptible group.  Thus the fraction of susceptibles using normal behavior tends to decrease.  

The second summand represents the the rate of change of $x$ due to imitiation dynamics.  The rate at which susceptibles using different behaviors encounter each other is proportional to $x(1-x)$.  The difference in payoffs of the two behaviors, given the current level of $I$, is
$$
p_n-p_a = k - (m_n-m_a)I.
$$
When this number is positive, normal behavior yields a larger payoff, so $x$ increases at a rate proportional to the difference between the payoffs; when this number is negative, $x$ decreases in the same manner.

The rate constant multiplying this summand is written as $\frac{1}{\epsilon}$ with $\epsilon>0$.  We will assume that this constant is large, so that $\epsilon$ is small.  In other words, we assume that behavior can change on a much faster time scale than that of the epidemic itself.

\section{Slow-fast structure\label{sec:slow-fast}}

\subsection{Slow-fast structure} System \eqref{P1}--\eqref{P3}, in which we recall that $\dot{\;} = \frac{d\;}{dt}$, is a slow-fast system \cite{jones94, kuehn15} with two slow variables, $S$ and $I$, and one fast variable, $x$; $t$ is the slow time.  Such systems are more commonly written with the last equation multiplied on both sides by $\epsilon$:
\begin{align}
\dot S &= -\big(\beta_nx+\beta_a(1-x)\big)SI, \label{Ps1} \\
\dot I &= \big(\beta_nx+\beta_a(1-x)\big)SI - \gamma I,\label{Ps2}  \\
\epsilon\dot x &=\epsilon x(1-x)(\beta_a-\beta_n)I+x(1-x)\big(k-(m_n-m_a)I\big). \label{Ps3} 
\end{align}

The fast time $\tau$ satisfies $t=\epsilon \tau$.  With  ${\;} ^\prime= \frac{d\;}{d\tau}$, system \eqref{Ps1}--\eqref{Ps3} becomes
\begin{align}
S^\prime &= -\epsilon\big(\beta_nx+\beta_a(1-x)\big)SI, \label{Pf1} \\
I^\prime &= \epsilon\big(\beta_nx+\beta_a(1-x)\big)SI - \epsilon\gamma I, \label{Pf2} \\
x^\prime &=\epsilon x(1-x)(\beta_a-\beta_n)I+x(1-x)\big(k-(m_n-m_a)I\big). \label{Pf3} 
\end{align}

The slow system \eqref{Ps1}--\eqref{Ps3}  and the fast system \eqref{Pf1}--\eqref{Pf3}  have the same phase portraits for $\epsilon>0$, but they have different limits at $\epsilon=0$.  For $\epsilon=0$, the slow system \eqref{Ps1}--\eqref{Ps3}  becomes the slow limit system
\begin{align}
\dot S &= -\big(\beta_nx+\beta_a(1-x)\big)SI, \label{Ps01} \\
\dot I &= \big(\beta_nx+\beta_a(1-x)\big)SI - \gamma I, \label{Ps02} \\
0 &=x(1-x)\big(k-(m_n-m_a)I\big), \label{Ps03} 
\end{align}
and the fast system \eqref{Pf1}--\eqref{Pf3}  becomes the fast limit system
\begin{align}
S^\prime &= 0, \label{Pf01} \\
I^\prime &= 0, \label{Pf02} \\
x^\prime &=x(1-x)\big(k-(m_n-m_a)I\big). \label{Pf03} 
\end{align}

Singular solutions are constructed by combining solutions of the slow limit system \eqref{Ps01}--\eqref{Ps03}  and the fast limit system \eqref{Pf01}--\eqref{Pf03}.  In many situations, solutions for small $\epsilon>0$ are close to singular solutions. 

\subsection{Fast limit system} For the fast limit system \eqref{Pf01}--\eqref{Pf03}, each line segment $(S,I)=(S_0,I_0)$ is invariant, and the triangles $x=0$ and $x=1$ consist of equilibria.  (The plane $I=\frac{k}{m_n-m_a}$ also consists of equilibria, but we will not make direct use of them.)  On line segments $(S,I)=(S_0,I_0)$ with $0\le I_0<\frac{k}{m_n-m_a}$, $\dot x>0$, so the solution $x(t)$ of  \eqref{Pf03} satisfies $\lim_{t\to-\infty}x(t)=0$ and $\lim_{t\to\infty}x(t)=1$.  On line segments $(S,I)=(S_0,I_0)$ with $\frac{k}{m_n-m_a}<I_0\le 1$, $\dot x<0$, so the solution $x(t)$ of \eqref{Pf03} satisfies $\lim_{t\to-\infty}x(t)=1$ and $\lim_{t\to\infty}x(t)=0$.  The fast dynamics just reflect the fact that normal behavior gives a higher payoff if $I<\frac{k}{m_n-m_a}$, and altered behavior gives a higher payoff if $I>\frac{k}{m_n-m_a}$.

Equilibria of the fast limit system \eqref{Pf01}--\eqref{Pf03} are normally attracting if $\frac{\partial \dot x}{\partial x}<0$ and normally repelling if $\frac{\partial \dot x}{\partial x}>0$. One can check that equilibria with $x=0$ are normally repelling for $I<\frac{k}{m_n-m_a}$ and normally attracting for $I>\frac{k}{m_n-m_a}$.   Equilibria with $x=1$ are the reverse.

\subsection{Slow limit system} The slow limit system \eqref{Ps01}--\eqref{Ps03} makes sense on the triangles $x=0$ and $x=1$.  

On the triangle $x=0$, the slow limit system reduces to
\begin{align}
\dot S &= -\beta_aSI,\label{Ps001} \\
\dot I &= \beta_aSI - \gamma I. \label{Ps002}
\end{align}
This is just an $SIR$ model with $\beta=\beta_a$ and basic transmission number $R_0<1$.  

Similarly, on the triangle
$x=1$, the slow limit system reduces to
\begin{align}
\dot S &= -\beta_nSI, \label{Ps011} \\
\dot I &= \beta_nSI - \gamma I. \label{Ps012}
\end{align}
This is just an $SIR$ model with $\beta=\beta_n$ and basic transmission number $R_0>1$. 

For $\epsilon>0$, the triangles $x=0$ and $x=1$ remain invariant.  The slow system \eqref{Ps1}--\eqref{Ps3}, restricted to $x=0$, remains \eqref{Ps001}-- \eqref{Ps002}.  Restricted to $x=1$ it remains \eqref{Ps011}-- \eqref{Ps012}. Thus the line segments $\{(S,I,x):0\le S\le1, \; I=0, x=0\}$ and $\{(S,I,x):0\le S\le1, \; I=0, x=1\}$ remain equilibria.

We will use the following notation where convenient:
\begin{itemize}
\item $\phi^\epsilon\big((S_0,I_0,x_0),t\big) = $ solution of \eqref{Pf1}--\eqref{Pf3} with $\phi^\epsilon\big((S_0,I_0,x_0),0\big) = (S_0,I_0,x_0)$.
\item  $\psi_0\big((S_0,I_0),t\big) = $ solution of \eqref{Ps001}--\eqref{Ps002} with $\psi_0\big((S_0,I_0),0\big) = (S_0,I_0)$.
\item  $\psi_1\big((S_0,I_0),t\big) = $ solution of \eqref{Ps011}--\eqref{Ps012} with $\psi_1\big((S_0,I_0),0\big) = (S_0,I_0)$.
\end{itemize}

\subsection{Entry-exit function for the triangle $x=0$\label{sec:ee0}}  

In the triangle $x=0$, let $(S_0,I_0) \in  \hat{T}$ with $I_0>\frac{k}{m_n-m_a}$, so  $(S_0,I_0)$ lies in the attracting portion of the triangle.  Let $(S(t),I(t))=\psi_0\big((S_0,I_0),t\big)$, let $t_1>0$, and let $(S_1,I_1)=\big(S(t_1),I(t_1)\big)$.  The solution $\big(S(t),I(t)\big)$ traces out a curve $\Gamma$, which from \eqref{SIRorbit} has the equation
\begin{equation}
\label{0orbit}
I+S-\frac{\gamma}{\beta_a}\ln S=v_0, \quad v_0=I_0+S_0-\frac{\gamma}{\beta_a}\ln S_0 = I_1+S_1-\frac{\gamma}{\beta_a}\ln S_1.
\end{equation}

We define the entry-exit integral 
\begin{multline}
\label{I0}
\mathcal{I}_0\big((S_0,I_0),(S_1,I_1)\big) =  \int_0^{t_1} k-(m_n-m_a)I(t) \; dt 
\\ =
\int_{S_0}^{S_1} -\frac{k-(m_n-m_a)(v_0-S+\frac{\gamma}{\beta_a}\ln S)}{\beta_aS(v_0-S+\frac{\gamma}{\beta_a}\ln S)}  \; dS.
\end{multline}
The second integral follows from the first by making the substitutions $S=S(t), \; dS = -\beta_a S(t)I(t) \; dt$, and $I=v_0-S+\frac{\gamma}{\beta_a}\ln S$, which follows from \eqref{0orbit}.  It cannot be evaluated analytically, but is readily evaluated numerically.

The integrand of the first integral is negative when $I>\frac{k}{m_n-m_a}$ and positive when $I<\frac{k}{m_n-m_a}$.  The integral represents accumulated attraction to (resp. repulsion from) the plane $x=0$ where the integrand is negative (resp. positive).  

\begin{Prop}
\label{P:exactly}
For each point $(S_0,I_0)$ in $\hat{T}$ with $I_0>\frac{k}{m_n-m_a}$, there is exactly one $t_1>0$ such that  $(S_1,I_1)=(S(t_1),I(t_1))$ satisfies \newline $\mathcal{I}_0\big((S_0,I_0),(S_1,I_1)\big)=0$.
\end{Prop}

Of course, $(S_1,I_1)$ lies in the region $I<\frac{k}{m_n-m_a}$.  Intuitively, at $(S_1,I_1)$ the accumulated repulsion from the plane $x=0$ balances the accumulated attraction to the plane.   We shall see that for small $\epsilon>0$, a solution of \eqref{Pf1}--\eqref{Pf3} that enters a neighborhood of the plane $x=0$ near $(S_0,I_0)$ will track a solution of \eqref{Ps001}--\eqref{Ps002} near $(S(t),I(t))$ until it leaves the neighborhood near $(S_1,I_1)$.  See Figure \ref{fig:entryexit} and Subsection \ref{sec:ee02}.

\begin{figure}[htb]
\includegraphics[width=3in]{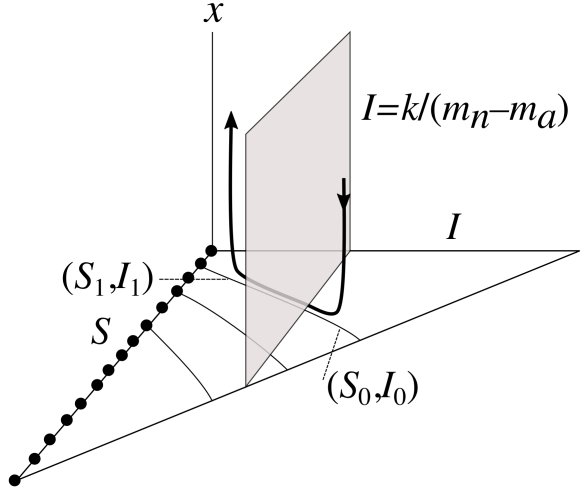}
\caption{A solution of \eqref{Pf1}--\eqref{Pf3} approaches the triangle $x=0$ near a point $(S_0,I_0)$, follows the solution of \eqref{Ps001}--\eqref{Ps002} through $(S_0,I_0)$ until $\mathcal{I}_0\big((S_0,I_0),(S_1,I_1)\big)=0$, then leaves the triangle.}
\label{fig:entryexit}
\end{figure}

\subsection{Entry-exit function for the triangle $x=1$\label{sec:ee1}} 

In the triangle $x=1$, let $(S_0,I_0) \in \hat{T}$ with $I_0<\frac{k}{m_n-m_a}$, so  $(S_0,I_0)$ lies in the attracting portion of the triangle.   Let $(S(t),I(t))=\psi_1\big((S_0,I_0),t\big)$, let $t_1>0$, and let $(S_1,I_1)=\big(S(t_1),I(t_1)\big)$.  The solution $\big(S(t),I(t)\big)$ traces out a curve $\Gamma$, which from \eqref{SIRorbit} has the equation
\begin{equation}
\label{0orbit}
I+S-\frac{\gamma}{\beta_n}\ln S=v_0, \quad v_0=I_0+S_0-\frac{\gamma}{\beta_n}\ln S_0 = I_1+S_1-\frac{\gamma}{\beta_n}\ln S_1.
\end{equation}

We define the entry-exit integral 
\begin{multline}
\label{I1}
\mathcal{I}_1\big((S_0,I_0),(S_1,I_1)\big) =  \int_0^{t_1} k-(m_n-m_a)I(t) \; dt 
\\ =
\int_{S_0}^{S_1} -\frac{k-(m_n-m_a)(v_0-S+\frac{\gamma}{\beta_n}\ln S)}{\beta_nS(v_0-S+\frac{\gamma}{\beta_n}\ln S)}  \; dS.
\end{multline}
The second integral follows from the first as in the previous subsection.

The integrand of the first integral  is negative when $I<\frac{k}{m_n-m_a}$ and positive when $I>\frac{k}{m_n-m_a}$.  The integral represents accumulated attraction to (resp. repulsion from) the plane $x=1$ where the integrand is negative (resp. positive).  

The system \eqref{Ps011}--\eqref{Ps012} on $T$ has a unique orbit  that is tangent to the line $I=\frac{k}{m_n-m_a}$.  The point of tangency is $(S^*,I^*)$, $I^*=\frac{k}{m_n-m_a}$.   Let  $\Gamma^*$ denote the part of this orbit with $S^* < S<1$.
Let $\Gamma^*$ have the equation $S=S_*(I)$, $0<I<\frac{k}{m_n-m_a}$.  Let $V_-=\{(S,I) \in \hat T : 0<S<S_*(I) \mbox{ and } I<\frac{k}{m_n-m_a}\}$, and let $V_+=\{(S,I) \in \hat T : S_*(I) \le S \mbox{ and } I<\frac{k}{m_n-m_a}\}$.  See Figure \ref{fig:V}.

\begin{figure}[htb]
\includegraphics[width=3in]{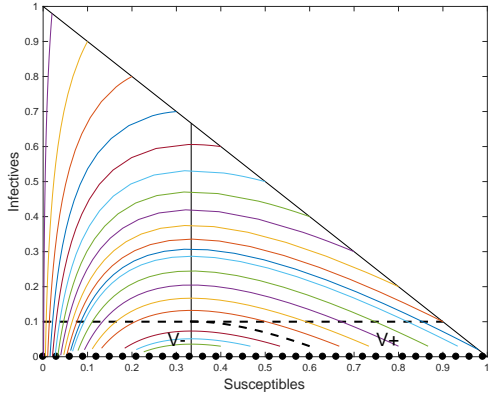}
\caption{Phase portrait of the fast limit system \eqref{Ps011}--\eqref{Ps012} in the triangle $x=1$ with $\beta_n = 1/2$ and $\gamma=1/6$.  The vertical line $S = \gamma/\beta = 1/3$ and the horizontal line $I=\frac{k}{m_n-m_a}=\frac{1}{10}$ are shown, as are the sets $V_-$ and $V_+$ bounded above by this line.  Solutions that start in $V_-$ approach equilibria without crossing the line $I=\frac{k}{m_n-m_a}$; solutions that start in $V_+$ cross the line.}
\label{fig:V}
\end{figure}

Let $(S_0,I_0) \in V_-$ and let $(S(t),I(t))=\psi_1\big((S_0,I_0),t\big)$.  Then \newline $\big(S(t),I(t)\big) \in V_-$ for all $t\ge0$.  Thus $\mathcal{I}_1\big((S_0,I_0),(S_1,I_1)\big)$ is never 0.  As $t\to\infty$, $\big(S(t),I(t)\big)$ approaches an equilibrium $(S_f,0)$ of \eqref{Ps011}--\eqref{Ps012}.  In this case,  for small $\epsilon>0$, a solution of \eqref{Pf1}--\eqref{Pf3}  that enters a neighborhood of the plane $x=1$ near $(S_0,I_0)$ will track a solution of \eqref{Ps011}--\eqref{Ps012} near $(S(t),I(t))$ and approach an equilibrium $(S_f^\epsilon,0,1)$ of \eqref{Pf1}--\eqref{Pf3} with $S_f^\epsilon$ near $S_f$.  See Subsection \ref{sec:lim}.

Let $(S_0,I_0) \in V_+$ and let $(S(t),I(t))=\psi_1\big((S_0,I_0),t\big)$.  Then  $\big(S(t),I(t)\big)$ enters the region $I\ge\frac{k}{m_n-m_a}$ at $t=t_{\rm{in}}>0$ and leaves that region at $t=t_{\rm{out}}\ge t=t_{\rm{in}}$.  

If $\int_0^{t_{\rm{out}}} k-(m_n-m_a)I(t) \; dt <0$, then there is no point $(S_1,I_1)$ where $\mathcal{I}_1\big((S_0,I_0),(S_1,I_1)\big)$.  As in the previous paragraph, let $(S_f,0) = \lim_{t\to\infty}\big(S(t),I(t)\big)$.  In this case also,  for small $\epsilon>0$, a solution of \eqref{Pf1}--\eqref{Pf3}  that enters a neighborhood of the plane $x=1$ near $(S_0,I_0)$ will track a solution of \eqref{Ps011}--\eqref{Ps012} near $(S(t),I(t))$ and approach an equilibrium $(S_f^\epsilon,0,1)$ of \eqref{Pf1}--\eqref{Pf3} with $S_f^\epsilon$ near $S_f$.

If $\int_0^{t_{\rm{out}}} k-(m_n-m_a)I(t) \; dt >0$, then there is a unique point $(S_1,I_1)$, with $I_1>\frac{k}{m_n-m_a}$,  where $\mathcal{I}_1\big((S_0,I_0),(S_1,I_1)\big)=0$.  For small $\epsilon>0$, a solution of \eqref{Pf1}--\eqref{Pf3} that enters a neighborhood of the plane $x=1$ near $(S_0,I_0)$ will track a solution of \eqref{Ps011}--\eqref{Ps012} near $(S(t),I(t))$ until it leaves the neighborhood near $(S_1,I_1)$.  See Subsection \ref{sec:ee12}.

\subsection{Singular orbits} Motivated by the previous subsections, we construct singular orbits of the system \eqref{Pf1}--\eqref{Pf3} (or equivalently \eqref{Ps1}--\eqref{Ps3}).  

Consider a starting point $(S_0,I_0,x_0)$ with $(S_0,I_0) \in \hat{T}$,  $I_0<\frac{k}{m_n-m_a}$, and $0<x_0<1$.  At this point, $\dot x>0$.

1. The first orbit in $\mathcal{S}$ is a fast orbit: the portion of the line $(S,I)=(S_0,I_0)$ with $x_0\le x <1$.

2. To describe the next orbit in $\mathcal{S}$, there are three cases.  Let $(S(t),I(t))=\psi_1\big((S_0,I_0),t\big)$, and,  given $t_1>0$, let $(S_1,I_1)=\big(S(t_1),I(t_1)\big)$.

2a. Suppose there exists $t_1>0$ such $\mathcal{I}_1\big((S_0,I_0),(S_1,I_1)\big)=0$ and $I_1>\frac{k}{m_n-m_a}$.  The next orbit of $\mathcal{S}$ is $\{(S(t),I(t)) : 0\le t \le t_1 \}$, a slow orbit. 

2b. Suppose there exists $t_1>0$ such that $\mathcal{I}_1\big((S_0,I_0),(S_1,I_1)\big)=0$, and $I_1=\frac{k}{m_n-m_a}$.  In this case the construction of the singular orbit fails.  (Notice $t_1 = t_{\rm{out}}$ from the previous subsection.)

2c. Suppose there is no $t_1>0$ such that $\mathcal{I}_1\big((S_0,I_0),(S_1,I_1)\big)=0$.  Then the next orbit of $\mathcal{S}$ is $\{(S(t),I(t)) : t\ge0 \}$, a slow orbit.  This orbit approaches an equilibrium of \eqref{Ps011}--\eqref{Ps012}.  The construction of $\mathcal{S}$ terminates.

3. We continue the construction of $\mathcal{S}$ in case 2a.  The next orbit in $\mathcal{S}$ is a fast orbit: the portion of the line $(S,I)=(S_1,I_1)$ with $0< x <1$.

4. Let $(S(t),I(t))=\psi_0\big((S_1,I_1),t\big)$.  By Proposition \ref{P:exactly} there is exactly one $t_1>0$ such that  $(S_2,I_2)=(S(t_1),I(t_1))$ satisfies $\mathcal{I}_0\big((S_1,I_1),(S_2,I_2)\big)=0$. We have $I_1<\frac{k}{m_n-m_a}$.  The next orbit of $\mathcal{S}$ is $\{(S(t),I(t)) : 0\le t \le t_1 \}$, a slow orbit. 

5. The next orbit in $\mathcal{S}$ is a fast orbit: the portion of the line $(S,I)=(S_2,I_2)$ with $0< x <1$.

We now continue the construction at step 2, setting $(S_0,I_0)$ equal to $(S_2,I_2)$.

Next we consider a starting point $(S_0,I_0,x_0)$ with $(S_0,I_0) \in \hat{T}$,  $I_0>\frac{k}{m_n-m_a}$, and $0<x_0<1$.  In this case the first orbit in $\mathcal{S}$ is again a fast orbit: the portion of the line $(S,I)=(S_0,I_0)$ with $0 < x \le x_0$.  We continue the construction of $\mathcal{S}$ at step 4 above, setting $(S_1,I_1)$ equal to $(S_0,I_0)$.  

In both cases the singular orbit $\mathcal{S}$ is an alternating sequence of fast and slow orbits, with the first orbit fast.  The slow orbits alternate between orbits in $x=0$ and orbits in $x=1$.  The last orbit is a slow orbit in $x=1$ that approaches an equilibrium, for which $I=0$.

\begin{Thm}
\label{th:smalleps}
Let $(S_0,I_0,x_0)$ satisfy $(S_0,I_0) \in \hat{T}$,  $I_0 \neq \frac{k}{m_n-m_a}$, and $0<x_0<1$.  Suppose the construction of the singular orbit $\mathcal{S}$ that starts at $(S_0,I_0,x_0)$ never fails at step 2 and terminates after a finite number of steps at $(S_f,0,1)$. 
Let $\Gamma^\epsilon$ denote the orbit of \eqref{Pf1}--\eqref{Pf3} that starts at  $(S_0,I_0,x_0)$. Then as $\epsilon \to 0$, $\Gamma^\epsilon \to \mathcal{S}$.  The terminal point $(S_f^\epsilon,0,1)$ of $\Gamma^\epsilon$ converges to $(S_f,0,1)$.
\end{Thm}

Roughly speaking, the fast jumps between $x=0$ and $x=1$ occur because the predominant behavior among the susceptibles has become less rewarding than the alternative.  When normal behavior predominates ($x$ near 1), if the fraction of infectives becomes high, behavior may switch to the altered form ($x$ near 0).   When altered behavior predominates ($x$ near 0), the fraction of infectives becomes low, and behavior swiches to the normal form ($x$ near 1).

However, the switch does not occur immediately when the number of infectives crosses the threshhold value $I=\frac{k}{m_n-m_a}$.  As was mentioned in the introduction, behavior changes because of encounters with other people whose behavior offers a higher payoff than one's own, so the current behavior is ``sticky."  The delay until behavior changes can be calculated in the limit from the entry-exit function.

\section{Examples\label{sec:ex}}

We consider the system \eqref{Ps1}--\eqref{Ps3} with the parameter values
$$
\beta_n=1/2, \; \beta_a=1/10,  \; \gamma=1/6,  \; k=3/10,  \; m_n=5,  \; m_a=2.
$$
From the values of $\beta_n$, $\beta_a$, and $\gamma$, we see $R_0=3$ for normal behavior and .6 for altered behavior.  The phase portrait of \eqref{Ps001}--\eqref{Ps002} in the triangle $x=0$ (resp. \eqref{Ps011}--\eqref{Ps012} in the triangle $x=1$) is given by Figure \ref{fig:SIR:a} (resp. Figure \ref{fig:SIR:b}).  The plane $I=\frac{k}{m_n-m_a}$ is $I=1/10$.

We shall consider singular orbits that start at $P_{\rm{start}}=(S_0,I_0,.98)$ with $I_0<1/10$.  Such a singular orbit starts with a fast solution from $P_{\rm{start}}$ to $(S_0,I_0,1)$.  One possibility is that the singular orbit immediately ends with an orbit of \eqref{Ps011}--\eqref{Ps012} from $(S_0,I_0,1)$ to a point $P_{\rm{end}}=(S_f,0,1)$; we would represent such a singular orbit by the sequence $(P_{\rm{start}}, P_{\rm{end}})$.  Otherwise we represent the singular orbit by a sequence 
$$
(P_{\rm{start}}, P_1, P_2, \ldots, P_{2k},P_{\rm{end}}), 
$$
where
\begin{itemize}
\item the first fast orbit goes from $P_{\rm{start}}=(S_0,I_0,.98)$ to  $(S_0,I_0,1)$;
\item the first slow orbit goes from $(S_0,I_0,1)$ to $P_1=(S_1,I_1,1)$;
\item the second fast orbit goes from $P_1=(S_1,I_1,1)$ to $(S_1,I_1,0)$;
\item the second slow orbit goes from $(S_1,I_1,0)$ to $P_2=(S_2,I_2,0)$ (unless it's the last slow orbit, see below);
\item the third fast orbit goes from $P_2=(S_2,I_2,0)$ to  $(S_2,I_2,1)$;
$$
\vdots
$$
\item the last fast orbit goes from $P_{2k}=(S_{2k},I_{2k},0)$ to  $(S_{2k},I_{2k},1)$;
\item the last slow orbit goes from $(S_{2k},I_{2k},1)$ to $P_{\rm{end}}=(S_f,0,1)$.
\end{itemize}
In other words, $P_1$, \ldots, $P_{2k}$ are the starting points of fast jumps;  $P_i$ with $i$ odd is in $x=1$, and  $P_i$ with $i$ even is in $x=0$.

Using the Matlab routines in the appendix, one can compute singular orbits for this system.  We give three examples.  Corresponding to each example we show the solution of the fast system \eqref{Pf1}--\eqref{Pf3} with the same starting point and $\epsilon = .005$, on the interval $0\le t\le 20,000$, computed using the Matlab ODE solver ode23s with the options RelTol=1e-10 and AbsTol=1e-11.  Because $\epsilon = .005$, 1000 units of fast time correspond to five units of slow time, i.e., five days.  To compare with the singular orbits, we give the value of $I$ at $x=1/2$ along each jump, and the value of $S$ at $t=30,000$.

\begin{example} A singular solution with two jumps.  
\begin{align*}
P_{\rm{start}} &= (.97,.03,.98)   \nonumber \\
P_1 &=  (.6713533014,  .2059798507,  1) \nonumber \\
P_2 &=  (.5714338970, .0373213930,  0) \nonumber \\
P_{\rm{end}} &=  (.1400580768,                                   0,     1) \nonumber 
\end{align*}
For the computed solution, jumps in $x$ occur successively at  $I=.20535$ and $I=.03735$; $S=.14017$ at $t=30,000$.  See Figure \ref{fig:ex1}.  Infections initially rise, then the epidemic is controlled by altered behavior for a while.  When the population switches back to normal behavior, infections rise for a while, then fall to zero.

\end{example}

\begin{example} A singular solution with four jumps.  This example was shown in the introduction.
\begin{align*}
P_{\rm{start}} &= (.96,.04,.98) \nonumber \\
P_1 &=  (.7197479246,  .1842413292,  1)  \nonumber \\
P_2 &=  (.6258761345,  .0451988482,  0)  \nonumber \\
P_3 &=  (.2763360357,  .1222273345,  1)  \nonumber \\
P_4 &=  (.2682106618,  .0806111708,  0) \nonumber \\
P_{\rm{end}} &=  (.1459222576,                                   0,     1)  \nonumber 
\end{align*}
For the computed solution, jumps in $x$ occur successively at  $I=.17815$, $I=.04756$, $I=.13366$, and $I=.07233$; $S=.15656$ at $t=30,000$.  In this example, when infections rise after the population switches back to normal behavior, the population again switches to altered behavior.  Eventually it switches back to normal behavior; this time there is no rise in infections, and infections fall to zero.
\end{example}

\begin{example} A singular solution with six jumps.
\begin{align*}
P_{\rm{start}} &= (.93,.07,.98) \nonumber \\
P_1 &=  ( .8251362461,  .1349850649,  1)  \nonumber \\
P_2 &=  ( .7667769297,  .0710900559,  0)  \nonumber \\
P_3 &=  ( .6515152002,  .1320533864,  1)  \nonumber \\
P_4 &=  ( .6155232547,  .0733319974,  0)  \nonumber \\
P_5 &=  ( .4804269385,  .1258291260,  1)  \nonumber \\
P_6 &=  ( .4615380487,  .0778669034,  0)  \nonumber \\
P_{\rm{end}} &=  ( .1387323862,                                   0,     1)  \nonumber 
\end{align*}
See Figure \ref{fig:ex1}.  For the computed solution, jumps in $x$ occur successively at  $I=.13931$, $I= .07344$, $I=.12876$, $I=.07547$, $I=.12329$, and $I=.07959$; $S=.13760$ at $t=30,000$.  In this example, the population switches to altered behavior three times after a rise in infections with normal behavior.  After the final episode of altered behavior, when the population switches back to normal behavior, infections rise and then fall to zero.
\end{example}

\begin{figure}[htb]
\centering
\subfigure[]
{
\label{fig:ex1:a}
\includegraphics[width=4in]{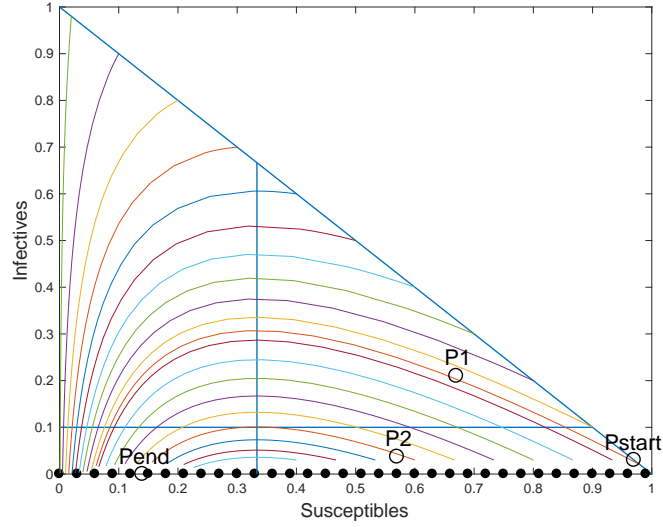}
}
\hspace{.0001in}
\subfigure[]
{
\label{fig:ex1:b}
\includegraphics[width=4in]{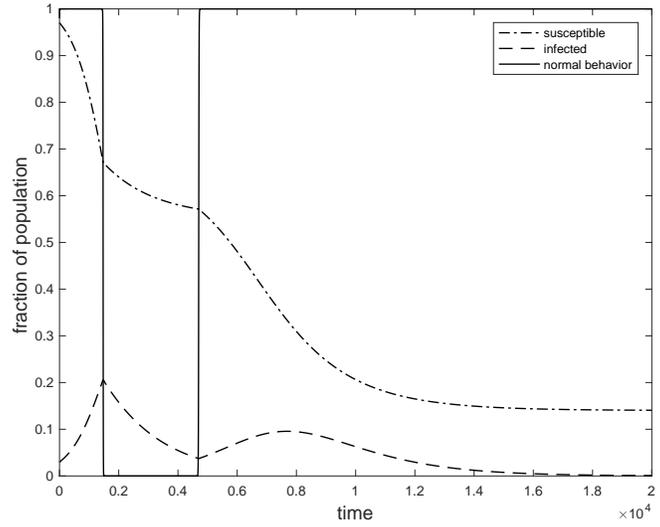}
}
\caption{Example 1, $P_{\rm{start}}=(.97,.03,.98)$.  (a) Phase portrait of the slow limit system in the triangle $x=1$, with the vertical line $S=\frac{\beta_n}{\gamma}=\frac{1}{3}$ and the horizontal line $I=\frac{k}{m_n-m_a}=\frac{1}{10}$ shown.  The slow orbits from $P_{\rm{start}}$ to $P_1$ and from $P_2$ to $P_{\rm{end}}$ are shown in this phase portrait.  The slow orbit from $P_1$ to $P_2$ lies in the triangle $x=0$; see Figure \ref{fig:SIR:a}. (b) Solution of the fast system \eqref{Pf1}--\eqref{Pf3} with the same starting point and $\epsilon=.005$. }
\label{fig:ex1}
\end{figure}

\begin{figure}[htb]
\centering
\subfigure[]
{
\label{fig:ex2:a}
\includegraphics[width=4in]{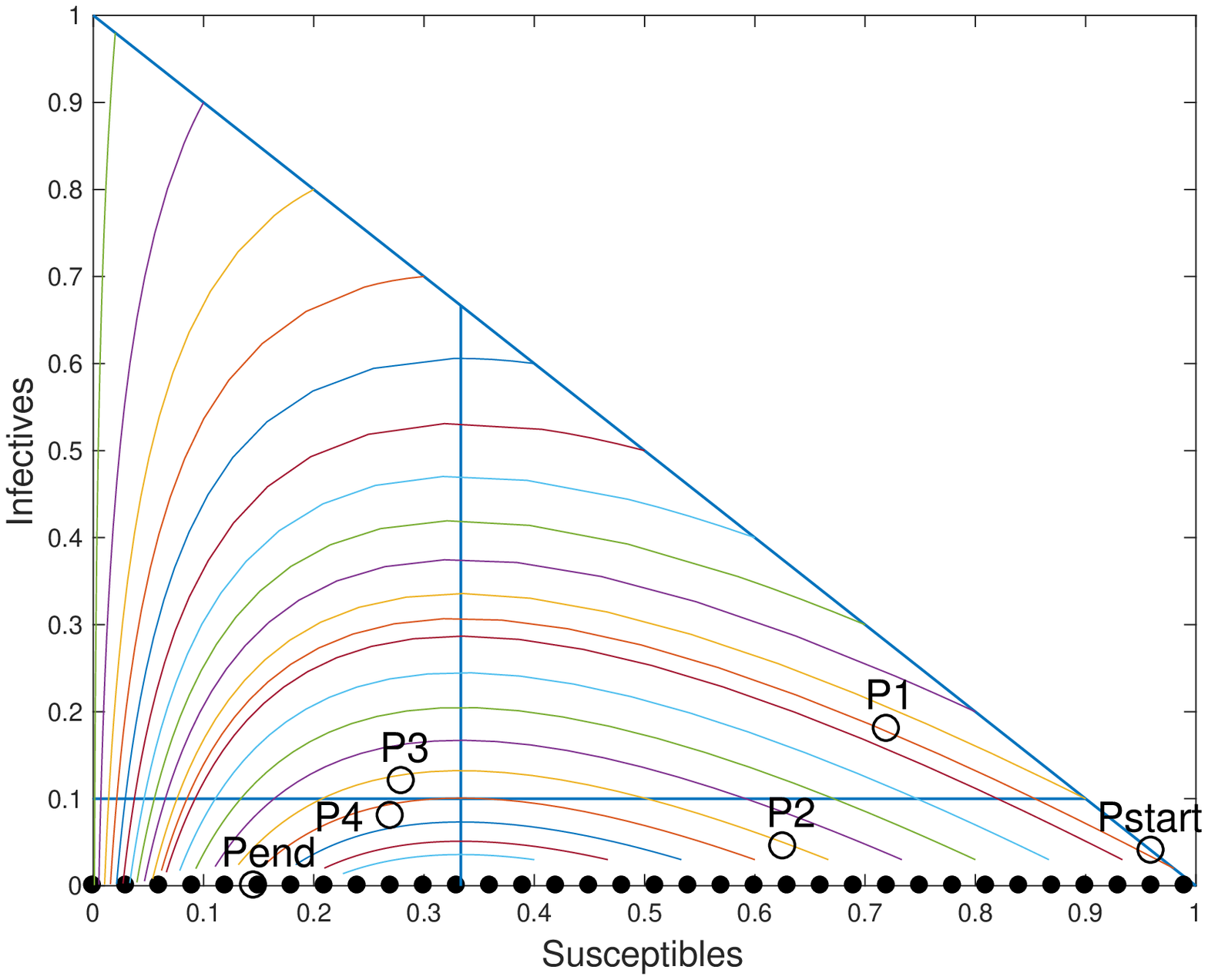}
}
\hspace{.0001in}
\subfigure[]
{
\label{fig:ex2:b}
\includegraphics[width=4in]{figs/timeplot2.eps}
}
\caption{Example 2, $P_{\rm{start}}=(.96,.04,.98)$. (a) Phase portrait of the slow limit system in the triangle $x=1$, with the vertical line $S=\frac{\beta_n}{\gamma}=\frac{1}{3}$ and the horizontal line $I=\frac{k}{m_n-m_a}=\frac{1}{10}$ shown.  The slow orbits from $P_{\rm{start}}$ to $P_1$, from $P_2$ to $P_3$, and from $P_4$ to $P_{\rm{end}}$ are shown in this phase portrait.  The slow orbits from $P_1$ to $P_2$ and from from $P_3$ to $P_5$ lie in the triangle $x=0$; see Figure \ref{fig:SIR:a}. (b) Solution of the fast system \eqref{Pf1}--\eqref{Pf3} with the same starting point and $\epsilon=.005$. }
\label{fig:ex1}
\end{figure}

\begin{figure}[htb]
\centering
\subfigure[]
{
\label{fig:ex3:a}
\includegraphics[width=4in]{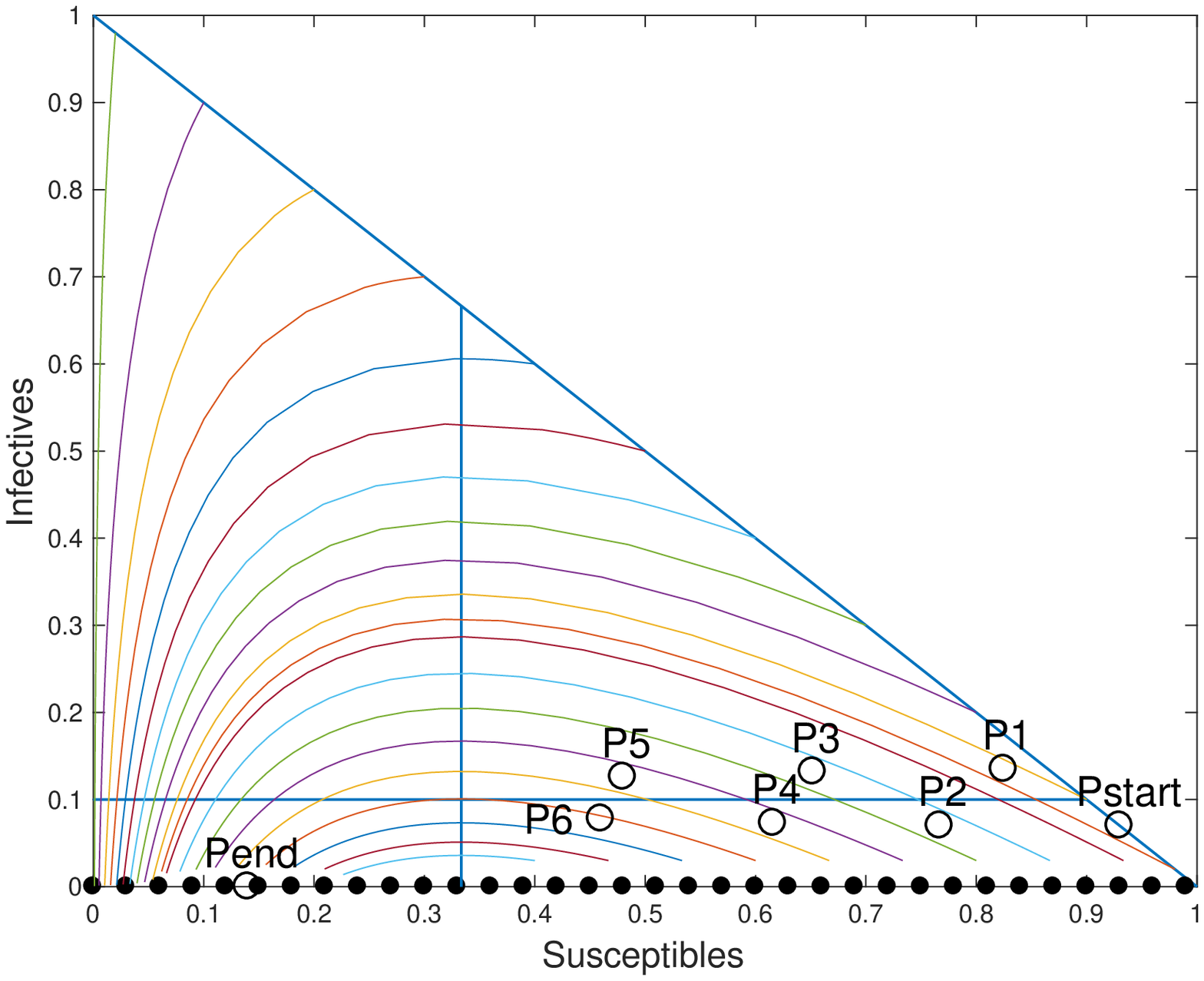}
}
\hspace{.0001in}
\subfigure[]
{
\label{fig:ex3:b}
\includegraphics[width=4in]{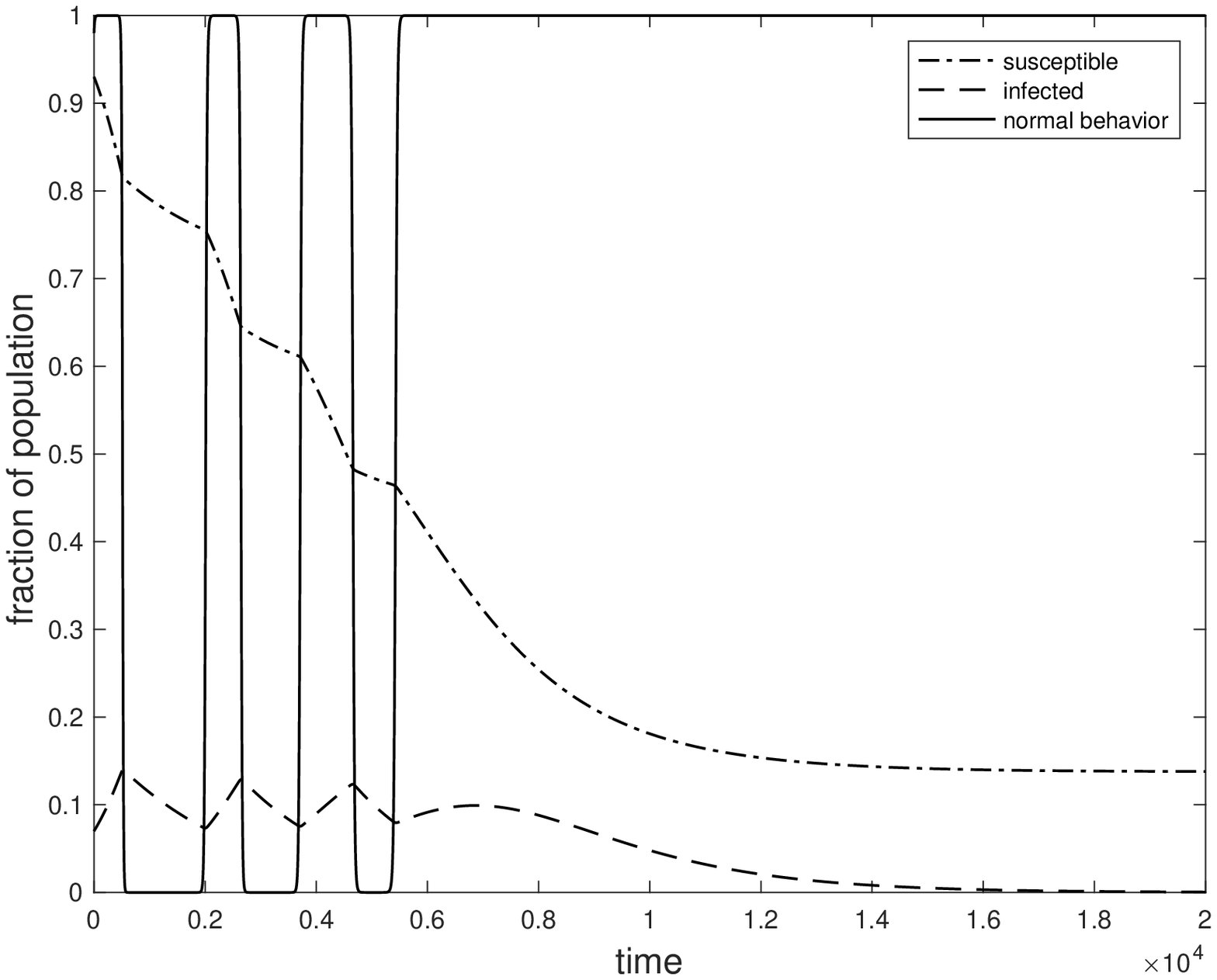}
}
\caption{Example 3,  $P_{\rm{start}}=(.93,.07,.98)$. (a) Phase portrait of the slow limit system in the triangle $x=1$, with the vertical line $S=\frac{\beta_n}{\gamma}=\frac{1}{3}$ and the horizontal line $I=\frac{k}{m_n-m_a}=\frac{1}{10}$ shown.  The slow orbits from $P_{\rm{start}}$ to $P_1$, from $P_2$ to $P_3$, from $P_4$ to $P_5$, and from $P_6$ to $P_{\rm{end}}$ are shown in this phase portrait.  The slow orbits from $P_1$ to $P_2$, from $P_3$ to $P_4$, and from $P_5$ to $P_6$ lie in the triangle $x=0$; see Figure \ref{fig:SIR:a}. (b) Solution of the fast system \eqref{Pf1}--\eqref{Pf3} with the same starting point and $\epsilon=.005$. }
\label{fig:ex1}
\end{figure}

\section{Proofs\label{sec:proofs}}

\subsection{Entry-exit function}

Let $U$ be an open subset of $\mathbb{R}^{n}$, $n\ge1$, and consider the system 
\begin{align}
c^\prime & = \epsilon p(c,z,\epsilon), \label{ee21}  \\
z^\prime & = zq(c,z,\epsilon), \label{ee22}
\end{align}
with $(c,z,\epsilon) \in U \times [0,z_0) \times [0,\epsilon_0)$ and $\;^\prime = \frac{d\;}{d\tau}$.   We assume 
\begin{enumerate}
\item[(E1)] $p$ and $q$ are of class $C^r$, $r\ge2$; 
\item[(E2)] if $q(c,0,0)=0$, then $Dq(c,0,0)p(c,0,0) > 0$.
\end{enumerate}
Assumption (E2) implies that the equation $q(c,0,0)=0$ defines a $C^r$ codimension-one submanifold $S$ of $U$.

Let $t=\epsilon\tau$ and let $\dot{\;}=\frac{d\;}{dt}$.  For $c_0 \in U$ with $q(c_0,0,0)<0$, let $\phi(c_0,t)$ denote the solution of $\dot c = p(c,0,0)$ with $\phi(c_0,0)=c_0$.  Assumption (E2) implies that $\phi(c_0,t)$ crosses the manifold $S$ at most once.

Given $t_1>0$, let $c_1 = \psi(c_0,t_1)$.  Define the
entry-exit integral
\begin{equation}
\label{I}
\mathcal{I}(c_0,c_1) = \int_{0}^{t_1} q(\psi(c_0,t),0,0) \; dt.
\end{equation}

\begin{Thm} 
\label{th:ee}
For system \eqref{ee21}--\eqref{ee22} satisfying (E1)--(E2), assume \newline $\mathcal{I}(\bar c_0,\bar c_1)=0$.  For a small neighborhood $\tilde U$ of $\bar c_0$ in $U$, define the entry-exit function $P^0: \tilde U \to U$ by $P^0(c_0)= c_1$, where $c_1$ is defined implicitly by $\mathcal{I}(c_0,c_1)=0$.  Fix $\delta>0$ sufficiently small.  For a given $\epsilon>0$, consider the solution of \eqref{ee21}--\eqref{ee22} that starts at $(c,z)=(c_0,\delta)$, with $c_0 \in \tilde U$.  Then: 
\begin{enumerate}
\item For  $\epsilon>0$ sufficiently small, the solution  reintersects the section $z=\delta$ at a point $(c,z) = (P^\epsilon(c_0),\delta)$. 
\item $P^\epsilon$ and $P^0$ are $C^r$  functions, and $P^\epsilon \to P^0$ in the $C^r$ sense as $\epsilon\to0$. 
\item Let $\Gamma^\epsilon$ denote the orbit of \eqref{ee21}--\eqref{ee22} from $(c_0,\delta)$ to $(P^\epsilon(c_0),\delta)$.  As $\epsilon\to0$, $\Gamma^\epsilon$ approaches the singular orbit of \eqref{ee21}--\eqref{ee22} consisting of
\begin{enumerate}
\item the line segment $[(c_0,\delta),(c_0,0))$;
\item $(\Gamma,0)$, where $\Gamma$ is the orbit of $\dot c = p(c,0,0)$ from $c_0$ to $c_1=P^0(c_0)$;
\item the line segment $\big((c_1,0),(c_1,\delta]\big)$.
\end{enumerate} 
\end{enumerate}
\end{Thm}

Theorem \ref{th:ee} is proved in \cite{demaesschalck08} under the assumption that for $\epsilon=0$, the system \eqref{ee21}--\eqref{ee22} has been written in a standard form.  The relation of the standard form to the form \eqref{ee21}--\eqref{ee22} is explained in \cite{weishi00}.

\subsection{Entry-exit function for the Poletti model in the triangle $x=0$ \label{sec:ee02}} 

The Poletti model \eqref{Pf1}--\eqref{Pf3} satisfies the hypotheses of Theorem \ref{th:ee} for $n=1$ and any $r\ge2$, with $(S,I)$ corresponding to $p$ and $x$ corresponding to $z$.  The set $U$ is $\hat T$, and $S$ is the line $I=\frac{k}{m_n-m_a}$.  ($\hat T$ is not open, since it includes a segment of the line $S+T=1$, but this does not cause any difficulty.)

Let $(S_0,I_0) \in \hat T$ with $I_0>\frac{k}{m_n-m_a}$. Let $(S(t),I(t))=\psi_0\big((S_0,I_0),t\big)$, let $t_1>0$, and let $(S_1,I_1)=(S(t_1),I(t_1))$. The formula \eqref{I0} for the entry-exit integral $\mathcal{I}_0$ follows immediately from \eqref{I}.

Proof of Proposition \ref{P:exactly}: We consider the solution $(S(t),I(t))$ of \eqref{Ps001}--\eqref{Ps002} defined above.  Since $I(t)$ is decreasing, there is a unique $t_*>0$ such that $I(t_*)=\frac{k}{m_n-m_a}$. The integral $\int_{0}^{t_1} k-(m_n-m_a)I(t) \; dt$ is negative and decreasing for $0<t\le t_*$ and is increasing for $t>t_*$.  To prove Proposition \ref{P:exactly} it suffices to show that $\int_{0}^{\infty} k-(m_n-m_a)I(t) \; dt=\infty$.  Since $\int_{0}^{\infty} k\; dt=\infty$, it suffices to show that $\int_{0}^{\infty} (m_n-m_a)I(t) \; dt$ is finite.  To see this, just note that $(S(t),I(t))$ approaches a normally attracting equilibrium $(S_f,0)$, so $I(t)\to0$ exponentially.  

\subsection{Entry-exit function for the Poletti model in the triangle $x=1$\label{sec:ee12}}

To treat the Poletti model \eqref{Pf1}--\eqref{Pf3} near $x=1$, we first make the change of variables $y=1-x$.  We obtain the system
\begin{align}
S^\prime &= -\epsilon\left(\beta_n(1-y)+\beta_ay\right)SI, \label{Pf1y} \\
I^\prime &= \epsilon\left(\beta_n(1-y)+\beta_ay\right)SI - \epsilon\gamma I, \label{Pf2y} \\
y^\prime &=\epsilon(1-y)y(\beta_n-\beta_a)I-(1-y)y\left(k-(m_n-m_a)I\right). \label{Pf3y} 
\end{align}

Define the curve $C$ to be the union of the line $I=\frac{k}{m_n-m_a}$, $0<S\le S^*$, and $\Gamma^*$ defined in Subsection \ref{sec:ee1}.  Let $U$ be the part of $T$ above $C$.

The system \eqref{Pf1}--\eqref{Pf3} satisfies the hypotheses of Theorem \ref{th:ee} for $n=1$ and any $r\ge2$, with $(S,I)$ corresponding to $p$ and $y$ corresponding to $z$.  The set $U$ is defined above, and $S$ is the line segment $I=\frac{k}{m_n-m_a}$, $S^*<S<1$.  (Again the set $U$ is not open because it includes a segment of the line $S+T=1$, but this does not cause any difficulty.)

As in the previous subsection, the formula for the entry-exit integral $\mathcal{I}_1\big((S_0,I_0),(S_1,I_1)\big)$ follows immediately from \eqref{I}.

\subsection{Solutions that approach equilibria\label{sec:lim}}

Recall the sets  $V_-$ and $V_+$ defined in Subsection \ref{sec:ee1}.   

\begin{Prop} \label{P:lim1}
Let $K$ be a compact subset of $V_-$.  For each $(S_0,I_0) \in K$, let $(S_f,0)=\lim_{t\to\infty}\psi_0\big((S_0,I_0),t\big)$.  Define $Q^0 : K \to \mathbb{R}$ by $Q^0(S_0,I_0) = S_f$.  Let $\delta>0$ be small. Then:
\begin{enumerate}
\item For small $\epsilon>0$ and for each $(S_0,I_0) \in K$, there exists $S_f^\epsilon \in \mathbb{R}$ such that $\lim_{t\to\infty}\phi^\epsilon\big((S_0,I_0,\delta),t\big)= (S_f^\epsilon,0,1)$.  
\item Define $Q^\epsilon: K \to \mathbb{R}$ by $Q^\epsilon(S_0,I_0) = S_f^\epsilon$. Then $Q^\epsilon$ and  $Q^0$ are $C^{r-1}$  functions, and $Q^\epsilon \to Q^0$ in the $C^{r-1}$ sense as $\epsilon\to0$. 
\end{enumerate}
\end{Prop}

\begin{proof}
Let $\hat K$ be a compact subset of $V_-$ that contains $K$ in its interior.  Let $\tilde K$ denote the union of $\hat K$, solutions of \eqref{Ps011}--\eqref{Ps012} that start in $K$, and the limits of these solutions. 

For the system \eqref{Pf1}--\eqref{Pf3} with $\epsilon=0$, $\tilde K$ is a union of equilbria that is compact, normally hyperbolic, and normally attracting.  
The point $(S_0,I_0,\delta)$ is in the stable fiber $(S_0,I_0,1)$. 

For small $\epsilon>0$, the set $\tilde K$ remains normally hyperbolic and normally attracting.
$(S_0,I_0,\delta)$ is in the stable fiber of a point $(S^\epsilon,I^\epsilon,1)$ near $(S_0,I_0,1)$.  The slow system \eqref{Ps1}--\eqref{Ps3}, restricted to $x=1$, is still \eqref{Ps011}--\eqref{Ps012}.  The solution of \eqref{Ps011}--\eqref{Ps012} through $(S^\epsilon,I^\epsilon)$ lies near the solution of \eqref{Ps011}--\eqref{Ps012} through $(S_0,I_0)$.  

Given these observations, the proposition follows from the theory of normally hyperbolic invariant manifolds.
\end{proof}

For each $(S_0,I_0) \in V_+$, there exists 
$t_{\rm{in}}(S_0,I_0)>0$ and $t_{\rm{out}}(S_0,I_0)\ge t_{\rm{in}}(S_0,I_0)$ such that $\psi_1\big((S_0,I_0),t\big)$ enters the region $I \ge \frac{k}{m_n-m_a}$ at $t=t_{\rm{in}}(S_0,I_0)$ and leaves that region at $t=t_{\rm{out}}(S_0,I_0)$.  

\begin{Prop} \label{P:lim2}
Let $K$ be a compact subset of $V_+$.  Assume that for each $(S_0,I_0) \in K$, $\int_0^{t_{\rm{out}}(S_0,I_0)} k-(m_n-m_a)I(t) \; dt  <0$, where $I(t)$ is defined by $(S(t),I(t)) = \psi_1\big((S_0,I_0),t\big)$.  Choose $t_1>\rm{sup}_K t_{\rm{out}}(S_0,I_0)$.  Define $Q^0 : K \to V_-$ by $Q^0(S_0,I_0) = \psi_1\big((S_0,I_0),t_1\big)$. Let $\delta>0$ be small.  For small $\epsilon>0$,  define $Q^\epsilon(S_0,I_0)$ and $Z^\epsilon(S_0,I_0)$ by $(Q^\epsilon(S_0,I_0),Z^\epsilon(S_0,I_0)) = \phi^\epsilon\big((S_0,I_0,\delta),t_1\big)$. Then:
\begin{enumerate}
\item $Q^\epsilon$, $Z^\epsilon$ and  $Q^0$ are $C^{r}$  functions.
\item $Q^\epsilon \to Q^0$ in the $C^{r}$ sense. 
\item There exists $A>0$ such that  $Z^\epsilon \le \delta e ^{-At_1}$.
\end{enumerate}
\end{Prop} 

\begin{proof}
See Proposition 3 of \cite{demaesschalck08} and the remark that follows it.  The key assumption needed is that for $(S_0,I_0) \in K$ and $0<t_2\le t_1$,  $\int_0^{t_2} k-(m_n-m_a)I(t) \; dt  < 0$.
\end{proof}

\begin{Prop} \label{P:lim3}
Proposition \ref{P:lim1} also holds for a compact subset $K$ of $V_+$ that satisfies the assumption of Proposition \ref{P:lim2}.
\end{Prop}

\begin{proof}
Roughly speaking, we want to apply Proposition \ref{P:lim1} to the compact set $\phi_1(K,t_1)$ in $V_-$. However, corresponding to the point $(S_1,I_1)=\phi_1\big((S_0,I_0),t_1\big)$, we want to look not at the solution of \eqref{Pf1}--\eqref{Pf3} that starts at $(S_1,I_1,\delta)$, but at the solution that starts at $\phi^\epsilon\big((S_0,I_0,\delta),t_1\big)$.  This requires minor changes to Proposition \ref{P:lim1}.
\end{proof}

In the situation of Proposition \ref{P:lim1} or \ref{P:lim2}, we can also describe the limiting position of orbits.

\begin{Prop}
\label{P:lim4}
Let $K$ be a compact subset of $V_-$ that satisfies the assumption of Proposition \ref{P:lim1}, or a compact subset of $V_+$ that satisfies the assumption of Proposition \ref{P:lim2}.  Let $(S_0,I_0) \in K$.  Then there is an equilibrium $(S_f,0)$ of \eqref{Ps011}--\eqref{Ps012}  such that $\phi_1(S_0,I_0),t) \to (S_f,0)$ as $t\to\infty$.  Let $\Gamma^\epsilon$ denote the orbit of \eqref{Pf1}--\eqref{Pf3} that starts at $(S_0,I_0,\delta)$.  As $\epsilon\to0$, $\Gamma^\epsilon$ approaches the singular orbit of \eqref{Pf1}--\eqref{Pf3} consisting of
\begin{enumerate}
\item the line segment $[(S_0,I_0,\delta),(S_0,I_0,1))$;
\item $\{\phi_1(S_0,I_0),t) : t \ge 0\}$.
\end{enumerate} 
\end{Prop}

\subsection{Proof of Theorem \ref{th:smalleps}}

We will only consider one type of singular orbit; the proof for other types is similar.  Let $(S_0,I_0,x_0)$ satisfy $(S_0,I_0) \in \hat{T}$,  $I_0 > \frac{k}{m_n-m_a}$, and $0<x_0<1$.  We will consider the following singular orbit $\mathcal{S}$:
\begin{enumerate}
\item Fast orbit from $(S_0,I_0,x_0)$ to $(S_0,I_0,0)$.
\item Slow orbit $\Gamma_1$ from $(S_0,I_0,0)$ to $(S_1,I_1,0)$  with $I_1 < \frac{k}{m_n-m_a}$ and $\mathcal{I}_0\big((S_0,I_0),(S_1,I_1)\big)=0$.
\item Fast orbit  from $(S_1,I_1,0)$ to $(S_1,I_1,1)$.
\item Slow orbit $\Gamma_2$ from $(S_1,I_1,1)$ to $(S_f,0,1)$.
\end{enumerate}

For a small $\delta>0$, let $E_0 = \{(S,I,x) \in P : x=\delta\}$ and $E_1 = \{(S,I,x) \in P : x=1-\delta\}$.  For a small $\epsilon>0$, let $\Gamma^\epsilon$ be the orbit of \eqref{Pf1}--\eqref{Pf3} that starts at $(S_0,I_0,x_0)$.  We break $\Gamma^\epsilon$ into parts.
\begin{enumerate}
\item $\Gamma^\epsilon_1$ from $(S_0,I_0,x_0)$ to $(\tilde S_0,\tilde I_0,\delta)\in E_0$.
\item $\Gamma^\epsilon_2$ from $(\tilde S_0,\tilde I_0,\delta)$ to the next intersection with $E_0$ at
$(\tilde S_1,\tilde I_1,\delta)$.
\item $\Gamma^\epsilon_3$ from $(\tilde S_1,\tilde I_1,\delta)$ to $(\tilde{\tilde S}_1,\tilde{\tilde I}_1,1-\delta)\in E_1$.
\item  $\Gamma^\epsilon_4$ from $(\tilde{\tilde S}_1,\tilde{\tilde I}_1,1-\delta)\in E_1$ to an equilibrium $(\tilde S_f,0,1)$
\end{enumerate}
Then, as $\epsilon\to0$:
\begin{enumerate}
\item By considering the fast system \eqref{Pf1}--\eqref{Pf3} on $\delta \le x \le x_0$, we see that $\Gamma^\epsilon_1$ converges to the line segment $[S_0,I_0,x_0),(S_0,I_0,\delta)]$.
\item From Theorem \ref{th:ee}, $\Gamma^\epsilon_2$ converges to the union of the line segment $[(S_0,I_0,\delta),(S_0,I_0,0)$, the curve $\Gamma_1$, and the line segment $[(S_1,I_1,0),(S_1,I_1,\delta)]$.
\item By considering the fast system \eqref{Pf1}--\eqref{Pf3} on $\delta \le x \le 1-\delta$, we see $\Gamma_{\epsilon3}$ converges to the line segment $[(S_1,I_1,\delta),(S_1,I_1,1-\delta)]$.
\item From Proposition \ref{P:lim4}, $\Gamma_{\epsilon4}$ converges to the union of the line segment $[(S_1,I_1,1-\delta),(S_1,I_1,0)]$ and the curve $\Gamma_2$.  Moreover, the limiting equilibrium $(S_f^\epsilon,0,1)$ converges to $(S_f,0,1)$.
\end{enumerate}

\section{Discussion\label{sec:disc}}

One mathematical issue has been left hanging.  Theorem \ref{th:smalleps} applies only to singular orbits of finite length.  I suspect that all singular orbits have finite length, but have not been able to prove it.

The model discussed in this paper could be generalized in several tantalizing directions..  One is to replace the susceptible group by several subgroups with different payoff functions.  The groups could represent, for example, those with sufficient resources to survive staying home, or with the ability to work from home, and those who need to work outside the home.   A second direction, suggested by the covid-19 pandemic, is to replace the infective group by subgroups.  There could be a group that is infected, and infective, but so far asymptomatic, so unaware of being infective.  Those in this group would continue to use the behavior they used when susceptible.  Some in this group would later become symptomatic; they would presumably change their behavior at this point.

\appendix
\section{Matlab routines}

The file findsingorbit.m is used to find a singular orbit.  Parameter values are entered in the file epimconstants.m.  The files entryexitint0.m and entryexitint1.m are used by findsingorbit.m to evaluate entry-exit integrals in $x=0$ and $x=1$ respectively.

\vspace{.3in}
\noindent\bf{epimconstants.m}
\begin{lstlisting}
% Constants used by other functions.

epsilon = 0.005;
betan = 0.5;
betaa = 0.1;
gam = 1/6;
k = 0.3;
mn = 5;
ma = 2;
\end{lstlisting}
\vspace{.3in}
\noindent\bf{entryexitint0.m}
\begin{lstlisting}
% Entry-exit function in the plane x=0.

function y = entryexitint0(S0,S1,v0)

% In the plane x=0, evaluates the integral from S0 to S1, with parameter v0,  
% that is used to define the entry-exit function.  Given S0 we will want to 
% find S1 such that the integral is 0.

syms S;
y = vpaintegral(top(S,v0)/bottom(S,v0),S0,S1);
return
end

function y1 = main(S,v0)
epimconstants;
syms S;
y1 = v0-S+(gam/betaa)*log(S);
return
end

function y2 = top(S,v0)
epimconstants;
syms S;
y2 = k-(mn-ma)*main(S,v0);
return
end

function y3 = bottom(S,v0);
epimconstants;
syms S;
y3 = -betaa*S*main(S,v0);
return
end
\end{lstlisting}
\vspace{.3in}
\noindent\bf{entryexitint1.m}
\begin{lstlisting}
function y = entryexitint1(S0,S1,v0)

% In the plane x=1, evaluates the integral from S0 to S1, with parameter v0,  
% that is used to define the entry-exit function.  Given S0 we will want to 
% find S1 such that the integral is 0.

syms S;
y = vpaintegral(top(S,v0)/bottom(S,v0),S0,S1);
return
end

function y1 = main(S,v0)
epimconstants;
syms S;
y1 = v0-S+(gam/betan)*log(S);
return
end

function y2 = top(S,v0)
epimconstants;
syms S;
y2 = k-(mn-ma)*main(S,v0);
return
end

function y3 = bottom(S,v0);
epimconstants;
syms S;
y3 = betan*S*main(S,v0);
return
end
\end{lstlisting}
\vspace{.3in}
\noindent\bf{findsingorbit.m}
\begin{lstlisting}
function singorbit = findsingorbit(S0,I0,x0)

epimconstants
syms u

singorbit = [S0 I0 x0];

while 0<S0 & 0<I0 & S0+I0<=1 & 0<x0 & x0<1 & I0~=k/(mn-ma)

    if I0>k/(mn-ma)
        disp('solution is attracted to x=0')
        % In this case the solution is attracted to the plane x=0 and will
        % follow a solution in that plane.   
        v0=I0+S0-(gam/betaa)*log(S0);

        % The solution will arrive at (S1,I1) with I1=0 where u=S1
        % satisfies:
        eqn1 = u-(gam/betaa)*log(u)==v0;
        S1=vpasolve(eqn1,u,[0,S0]);

        % The solution in SI-space will arrive at (S2,I2) with I2=k/(mn-ma) 
        % where u=S2 satisfies:
        eqn2 = u-(gam/betaa)*log(u)==v0-k/(mn-ma);
          
        % Find where the solution in SI-space crosses the line I2=k/(mn-ma).
        disp('solution leaves the plane x=0')
        S2=vpasolve(eqn2,u,[S1,S0]);
        syms uEnd;
        eqn3 = entryexitint0(S0,uEnd,v0)==0;
        u4=vpasolve(eqn3,uEnd,[S1,S2]);
        % Continue solution from following point.
        S0=u4;
        I0=v0-S0+(gam/betaa)*log(S0);
        x0=0.1;
        singorbit = [singorbit;[S0 I0 x0]]; 

    else
        disp('solution is attracted to the plane x=1')
        % In this case the solution is attracted to the plane x=1 and will
        % follow a solution in that plane.
        v0=I0+S0-(gam/betan)*log(S0);

        % The solution will arrive at (S1,I1) with I1=0 where u=S1
        % satisfies:
        eqn1 = u-(gam/betan)*log(u)==v0;
        S1=vpasolve(eqn1,u,[0,S0]);

        % The following is the value of v0 for which the integral curve in 
        % SI-space has its max on the line I=k/(mn-ma).
        v0tan = (gam/betan)-(gam/betan)*log(gam/betan)+k/(mn-ma);

        if v0 <= v0tan
            disp('entire solution lies below I=k/(mn-ma), solution terminates')
            % In this case the solution will not leave the plane x=1.
            S0=S1;
            I0=0;
            x0=1;
            singorbit = [singorbit;[S0 I0 x0]];

        else
            disp('solution crosses the line I=k/(mn-ma)')
            % In this case the solution in SI-space meets the line 
            % I=k/(mn-ma) in two points with S-values S3>S2.
            eqn2 = u-(gam/betan)*log(u)==v0-k/(mn-ma);
            S3=vpasolve(eqn2,u,[S1,gam/betan]);
            int2 = entryexitint1(S0,S3,v0);
            
            if int2<=0
                disp('nevertheless solution terminates')
                % In this case the solution will not leave the plane x=1.
                S0=S1;
                I0=0;
                x0=1;
                singorbit = [singorbit;[S0 I0 x0]];

            else 
                disp('solution leaves the plane x=1')
                S2=vpasolve(eqn2,u,[gam/betan,S0]);
                syms uEnd
                eqn3 = entryexitint1(S0,uEnd,v0)==0;
                u4=vpasolve(eqn3,uEnd,[S3,S2]);
                % Continue solution from following point.
                S0=u4;
                I0=v0-S0+(gam/betan)*log(S0);
                x0=0.9;
                singorbit = [singorbit;[S0 I0 x0]];
                
            end
            
        end
        
    end
    
end

return

end
\end{lstlisting}

\end{document}